\documentclass[sigconf,balance=false,screen=true]{acmart}
\usepackage{float}
\usepackage[T1]{fontenc}
\usepackage{changepage}
\usepackage[noEnd=true]{algpseudocodex}
\usepackage{fixltx2e}
\usepackage{booktabs}
\usepackage{xspace}
\usepackage[ruled]{algorithm}
\usepackage{graphicx}
\usepackage{mathtools}
\usepackage[shortlabels]{enumitem}
\usepackage{cleveref}
\crefname{pseudoline}{line}{lines}
\Crefname{pseudoline}{Line}{Lines}
\usepackage{xspace}
\usepackage{array}
\usepackage{pseudo}
\pseudoset{
  compact,
  label=\scriptsize\sffamily\color{gray}\arabic*,
  ctfont=\sffamily,
  ct-left=\quad\ctfont{[},
  ct-right=\ctfont{]},
}

\setlist[description]{parsep=0.1px}
\usepackage[shortlabels]{enumitem}

\DeclareMathOperator{\mon}{Mon}
\DeclareMathOperator{\lm}{lm}
\DeclareMathOperator{\NF}{nf}
\DeclareMathOperator{\nxt}{next}
\DeclareMathOperator{\Id}{Id}
\DeclareMathOperator{\cand}{cand}
\DeclareMathOperator{\lift}{lift}
\DeclareMathOperator{\result}{result}
\DeclareMathOperator{\gen}{gen}
\DeclareMathOperator{\homo}{hom}

\DeclareMathOperator{\rk}{rank}

\newcommand\prim{\mathfrak{p}}
\newcommand\field{\mathbb{K}}
\newcommand\NN{\mathbb{N}}
\newcommand\ZZ{\mathbb{Z}}
\newcommand\QQ{\mathbb{Q}}
\newcommand\FF{\mathbb{F}}

\newcommand\xx{\mathbf{x}}

\newcommand\zz{\mathbf{z}}
\newcommand\xz{\zz,\xx}
\newcommand\genI{I^{\gen}}
\newcommand\homI{I^{\homo}}
\newcommand\maxi{\mathfrak{m}}
\newcommand\mord{\prec}
\newcommand\mordx{\mord_{\xx}}
\newcommand\mordz{\mord_{\zz}}

\DeclareMathOperator\INput{in}
\newcommand\ormord{\prec_{\INput}}

\DeclareMathOperator\OUTput{out}
\newcommand\targmord{\prec_{\OUTput}}

\DeclareMathOperator\DRL{drl}
\newcommand\drl{\preceq_{\DRL}}
\newcommand\drll{\prec_{\DRL}}
\newcommand\drlg{\succ_{\DRL}}
\DeclareMathOperator\LEX{lex}

\newcommand\lexl{\prec_{\LEX}}

\newcommand\loc[2]{#1[#2^{-1}]}

\newcommand\setof[2]{\left\{#1 \;\middle|\; #2\right\}}
\newcommand\compl[1]{O\left(#1\right)}
\newcommand\qcompl[1]{\widetilde{O}\left(#1\right)}

\newcommand\rT{\mathrm{T}}
\newcommand{\titlenewline}{\textsc{\texorpdfstring{\\}{ }}\xspace}
\newcommand{\Ffour}{\textsc{\texorpdfstring{F\textsubscript{4}}{F4}}\xspace}
\newcommand\Fquatre{\Ffour}
\newcommand\ffour{\Ffour}
\newcommand{\Ffive}{\textsc{\texorpdfstring{F\textsubscript{5}}{F5}}\xspace}
\newcommand\Fcinq{\Ffive}

\newcommand\FGLM{\textsc{FGLM}\xspace}
\newcommand\mondeg[1]{\mathrm{m}_{#1}}

\newcommand\oscar{\texttt{OSCAR}\xspace}
\newcommand\msolve{\texttt{msolve}\xspace}
\newcommand\julia{\texttt{Julia}\xspace}

\newcommand\myvspace[1]{\vspace*{#1}}

\makeatletter
\def\bbordermatrix#1{\begingroup \m@th
  \@tempdima 4.75\p@
  \setbox\z@\vbox{%
    \def\cr{\crcr\noalign{\kern2\p@\global\let\cr\endline}}%
    \ialign{$##$\hfil\kern2\p@\kern\@tempdima&\thinspace\hfil$##$\hfil
      &&\quad\hfil$##$\hfil\crcr
      \omit\strut\hfil\crcr\noalign{\kern-\baselineskip}%
      #1\crcr\omit\strut\cr}}%
  \setbox\tw@\vbox{\unvcopy\z@\global\setbox\@ne\lastbox}%
  \setbox\tw@\hbox{\unhbox\@ne\unskip\global\setbox\@ne\lastbox}%
  \setbox\tw@\hbox{$\kern\wd\@ne\kern-\@tempdima\left[\kern-\wd\@ne
    \global\setbox\@ne\vbox{\box\@ne\kern2\p@}%
    \vcenter{\kern-\ht\@ne\unvbox\z@\kern-\baselineskip}\,\right]$}%
  \null\;\vbox{\kern\ht\@ne\box\tw@}\endgroup}
\makeatother

\title{Computing Generic Fibers of Polynomial Ideals\titlenewline
  with \FGLM and Hensel Lifting}

\author{J\'er\'emy Berthomieu}
\affiliation{%
	\institution{Sorbonne Universit\'e, \textsc{CNRS}, \textsc{LIP6}}
	\city{F-75005 Paris}
	\postcode{75252}\country{France}}
\email{jeremy.berthomieu@lip6.fr}

\author{Rafael Mohr}
\affiliation{%
	\institution{Sorbonne Universit\'e, \textsc{CNRS}, \textsc{LIP6}}
	\city{F-75005 Paris}
	\postcode{75252}\country{France}}
\affiliation{%
	\institution{Rheinland-Pfälzische Technische
          Universität Kaiserslautern-Landau, Fachbereich
          Mathematik}
	\city{G-67663 Kaiserslautern}
	\postcode{67663}\country{Germany}}
\email{rafael.mohr@lip6.fr}

\copyrightyear{20NN}
\acmYear{20NN}
\setcopyright{acmlicensed}\acmConference[ISSAC 'NN]{Proceedings of the 20NN International Symposium on Symbolic and Algebraic Computation}{July XX--YY, 20NN}{City, ZZ, USA}
\acmBooktitle{Proceedings of the 20NN International Symposium on Symbolic and Algebraic Computation (ISSAC 'NN), July XX--YY, 20NN, City, ZZ, USA}
\acmPrice{yy.00}
\acmDOI{10.xxxx/xxxxxxx.xxxxxxx}
\acmISBN{xxx-x-xxxx-xxxx-x/xx/xx}


\begin{document}

\fancyhead{}
\newtheorem{remark}[theorem]{Remark}

\begin{abstract}
  We describe a version of the FGLM algorithm that can be used to
  compute generic fibers of positive-dimensional polynomial ideals. It
  combines the FGLM algorithm with a Hensel lifting strategy. In
  analogy with Hensel lifting, we show that this algorithm has a
  complexity quasi-linear in the number of terms of certain
  $\maxi$-adic expansions we compute. Some provided experimental data
  also demonstrates the practical efficacy of our algorithm.
\end{abstract}
\thanks{%
  The authors are supported by the joint ANR-FWF
  ANR-19-CE48-0015 \textsc{ECARP}
  and
  ANR-22-CE91-0007 \textsc{EAGLES}
  projects,
  ANR-19-CE40-0018 \textsc{De Rerum Natura}
  project,
  DFG Sonderforschungsbereich TRR 195 project
  and
  grants
  DIMRFSI 2021-02–C21/1131 of the Paris Île-de-France Region,
  FA8665-20-1-7029 of
  the EOARD-AFOSR,
  and
  Forschungsinitiative Rheinland-Pfalz.
  We thank
  the referees for their valuable comments on the paper
  and
  Ch.~Eder, P.~Lairez, V.~Neiger and M.~Safey El Din for fruitful discussions.
}

\maketitle

\section{Introduction}
\label{sec:intro}

\paragraph{Scientific Context}

\emph{Gröbner bases} lie at the forefront of the algorithmic treatment
of polynomial systems and ideals in symbolic computation. They are
defined as special generating sets of polynomial ideals which allow to
decide the ideal membership problem via a multivariate version of
polynomial long division. Given a Gröbner basis for a polynomial
ideal, a lot of geometric and algebraic information about the
polynomial ideal at hand can be extracted, such as the degree,
dimension or Hilbert function. We refer to~\cite{becker1993}
for a comprehensive treatment of the subject.

Notably, Gröbner bases depend on two parameters: The polynomial ideal
which they generate and a \emph{monomial order}, i.e.\ a certain kind
of total order on the set of monomials of the underlying polynomial
ring.  Then, the geometric and ideal-theoretic information that can be
extracted from a Gröbner basis depends on the chosen
monomial order.  For example, \emph{elimination orders} allow, as the
name suggests, to eliminate a chosen subset of variables from the
given polynomial ideal (i.e.\ to project on an affine subspace in a
geometric sense).

While Gröbner bases for elimination orders are frequently of interest,
it has been observed that all algorithms to compute Gröbner bases
based on the famous Buchberger algorithm~\cite{buchberger1965},
such as \Fquatre~\cite{faugere1999} and
\Fcinq~\cite{faugere2002}, are substantially more
well-behaved when used with non-elimination orders (most notably,
the \emph{degree reverse lexicographical} $\drll$ order).

This has motivated the design of numerous \emph{change of order}
algorithms: The task is to convert a given Gröbner basis w.r.t.\ one
order into a Gröbner basis w.r.t.\ another
order.
We mention here the Hilbert-driven algorithm by~\cite{traverso1996},
the Gröbner walk algorithm by~\cite{collart1997a} and, most notably
for this paper, the \FGLM algorithm~\cite{faugere1993a} and its
variants~\cite{faugere2014a,faugere2017a,neiger2020,berthomieu2022b}.

Furthermore, most ideal-theoretic operations in commutative algebra
(such as saturation and intersection) can be performed using Gröbner
bases by writing down a certain ideal associated to the given
polynomial ideal, choosing a certain monomial order and computing a
Gröbner basis for this associated ideal. Here, Gröbner basis
computation is used as a black box. It has recently been observed,
partly by the authors of this paper, that it can be (sometimes
substantially) more efficient to design \emph{dedicated} Gröbner basis
algorithms for specific ideal-theoretic tasks,
see~\cite{berthomieu2023,eder2023}.

\myvspace{-1em}
\paragraph{Problem Statement \& Contributions}

This paper is concerned with the algorithmic treatment of the
following problem: Fix a polynomial ring $R\coloneqq\field[\xz]$ in
two finite sets of variables $\xx$ and $\zz$ over a field $\field$ and
an ideal $I$ in $R$. Assume that the map
$\varphi :\field[\zz]\rightarrow \field[\xz]/I$ is injective and has \emph{generically
  finite fiber}, i.e.\ assume that the the \emph{generic fiber}
$\genI\coloneqq I\cdot\field(\zz)[\xx]$ of $I$ is zero-dimensional. Given a
Gröbner basis of $I$ w.r.t.\ a monomial order $\ormord$, we want to
compute a Gröbner basis $G$ of $\genI$ w.r.t.\ another monomial order
$\targmord$. One key motivation to solve this problem is that, when
$\targmord$ is a suitable elimination order, the Gröbner basis $G$ can
be used to compute a primary decomposition of the ideal $I$ (or an
irreducible decomposition of the algebraic set defined by $I$), see
\cite{becker1993} for details. Being able to compute such
decompositions has numerous applications, we mention for example the
algorithm presented in \cite{helmer2023b} which uses primary
decompositions to compute so-called Whitney stratifications of
singular varieties.

The algorithm we design to solve this problem relates to the two
research directions previously mentioned: It is a \emph{dedicated}
algorithm to perform an ideal-theoretic operation (by computing a
representation of the generic fiber of a suitably chosen map) and it
performs a change of order (by going from $\ormord$ to
$\targmord$). Our proposed solution to this problem can be seen as a
combination of the previously mentioned \FGLM algorithm with classical
Hensel lifting techniques. More precisely, if we let
$\maxi \coloneqq \langle \zz \rangle$, then, under some assumptions which are detailed in
this paper, we will compute $G$ by computing its image in
$(\field[\zz]/\maxi)[\xx] \simeq \field[\xx]$ and then lifting it modulo
higher and higher powers of $\maxi$. This lifting step uses the same
core idea as the \FGLM algorithm. With this approach, we expect that
our algorithm can be transported without much difficulty to the
setting where a Gröbner basis $G$ of a zero-dimensional ideal in
$\QQ[\xx]$ is required: Given $p$ a well-chosen prime number and $k\in\NN^*$
sufficiently large,
it would extract $G$ from its image in $(\ZZ/p\ZZ)[\xx]$ and then lift it modulo
$p^k$.
We
show that, similar to classical Hensel lifting, our algorithm runs in
arithmetic complexity quasi-linear in the number of terms of degree at most
the precision up to which we
need to lift when a ``quadratic lifting strategy'' is chosen, see
\Cref{cor:algcompl}, which implies in particular the following

\begin{theorem}
  \label{th:main}
  Let $f_1,\ldots,f_c$ be generic polynomials of respective degrees
  $d_1,\ldots,d_c$ in $\field[\xz]$. Assume that the
  $\drll$-Gröbner basis of $I=\langle f_1,\ldots,f_c\rangle$ is known and that the
  $\targmord$-Gröbner basis $G$ of
  $I\cdot\field(\zz)[\xx]$ has coefficients which are
  rational functions with degrees at most $\delta$ in the numerators and
  denominators. Let $\mondeg{2\delta}$ be the number of monomials in
  $\zz$ up to degree $2\delta$. Then, one can compute $G$ up
  to precision $2\delta$ using
  $\qcompl{\mondeg{2\delta} c (d_1\cdots d_c)^3}$ operations in $\field$.
\end{theorem}
Note that knowing $G$ up to precision $2\delta$ is enough to recover
$G$ by means of Padé approximants, see \Cref{lem:pade}.

\paragraph{Related Work}

Gröbner bases of generic fibers, as defined in the previous paragraph,
are classically computed using \emph{block monomial orders}, see
e.g.~\cite[Lemma~8.93]{becker1993}. Besides that, morally similar to
our algorithm, there is a rich body of literature about multi-modular
Gröbner basis computations~\cite{arnold2003a, ebert1983, pauer1992,
  traverso1989a} and Hensel/modular lifting techniques for Gröbner
bases~\cite{winkler1988a, grabe1993, schost2023a}.

Outside of the world of Gröbner bases, there are other data structures
for algorithmically manipulating polynomial ideals, or the algebraic
sets defined by them, which encode polynomial ideals by their generic
fiber associated to a well-chosen projection. We mention in particular
\emph{geometric resolutions}, see e.g.~\cite{giusti2001,schost2003}, and
\emph{triangular sets}, see e.g.~\cite{hubert2003} for a survey.

Our work also relates to specialization results for Gröbner basis,
i.e.\ results on the question whether a Gröbner basis remains a Gröbner
basis after specializing some of the variables,
see~\cite{gianni1989a, kalkbrener1989, becker1994}.

\paragraph{Outline}

In 
\Cref{sec:prelim}, we give necessary preliminaries both on Gröbner
bases and on the needed commutative algebra to state and prove the
correctness of our algorithms in 
\Cref{sec:alg}. In 
\Cref{sec:compl}, we transport the complexity statements for the
original \FGLM algorithm to our setting.  Finally we give some
benchmarks for a \julia implementation of our main algorithm
in 
\Cref{sec:bench}, comparing it to computing generic fibers using just
elimination orders.

\myvspace{-0.5em}
\section{Preliminaries}
\label{sec:prelim}
\myvspace{-0.25em}
\subsection{Gröbner Bases}
\label{sec:gbs}

In order to be self-contained, 
we recall some definitions
and basic properties related to Gröbner bases of polynomial ideals.

For a set of variables $\xx\coloneqq\{x_1,\dots,x_n\}$, we denote by
$\mon(\xx)$ the set of monomials in $\xx$, and for a field
$\field$, we let $R\coloneqq\field[\xx]$ be the multivariate polynomial ring in
$\xx$ over $\field$.
\begin{definition}
  \label{def:monorder} A \emph{monomial order} $\mord$ on $\xx$ is a total
  order on $\mon(\xx)$
  \begin{enumerate}
  \item extending the partial order on $\mon(\xx)$ given by divisibility and
  \item compatible with multiplication, i.e.\ we have
    \begin{align*}
      u \mord v \; \Rightarrow \; wu \mord wv \quad \forall u,v,w\in \mon(\xx).
    \end{align*}
  \end{enumerate}
\end{definition}

Of importance for us is the \emph{degree reverse lexicographic} order:
\begin{definition}
  \label{def:drl} The \emph{degree reverse lexicographic $\drll$} order on
  $\mon(\xx)$ is defined as follows for $u,v\in \mon(\xx)$:
  $u \drll v$ iff $\deg u < \deg v$ or $\deg u = \deg v$ and the
  last nonzero exponent of $u/v$ is positive.
\end{definition}

We will also need the notion of a \emph{block order}:
\begin{definition}
  Let $\xx$ and $\zz$ be two finite sets of variables. Write each
  monomial $u\in \mon(\xx \cup \zz)$ uniquely as a product
  $u = u_{\xx}u_{\zz}$ with $u_{\xx}\in \mon(\xx)$ and
  $u_{\zz}\in \mon(\zz)$. Fix a monomial order $\mordx$ on
  $\mon(\xx)$ and a monomial order $\mordz$ on $\mon(\zz)$. The
  corresponding \emph{block order eliminating $\xx$} is defined as
  follows: $u \mord v$ iff $u_{\xx} \mordx v_{\xx}$ or $u_{\xx}=v_{\xx}$
  and $u_{\zz} \mordz v_{\zz}$ for $u,v\in \mon(\xx\cup \zz)$.
\end{definition}

A monomial order on $\xx$ yields a notion of \emph{leading monomial}
in $R$:

\begin{definition}
  \label{def:lm} Let $\mord$ be a monomial order on $\mon(\xx)$. For
  a nonzero element $f\in R$ the \emph{leading monomial} of $f$ w.r.t.\
  $\mord$, denoted $\lm_{\mord}(f)$, is the $\mord$-largest monomial
  in the support of $f$. For a finite set $F$ in $R$ we define
$\lm_{\mord}(F) \coloneqq \setof{\lm_{\mord}(f)}{f\in F}$.
  For an ideal $I$ in $R$ we define the
  \emph{leading monomial ideal} of $I$ as
  $\lm_{\mord}(I)\coloneqq\langle \lm_{\mord}(f)\;|\;f\in I\rangle$.
\end{definition}

Fixing a monomial order gives \emph{normal forms} for images of
elements in quotient rings of $R$:

\begin{definition}
  Let $I$ be an ideal in $R$ and let $\mord$ be a monomial order on
  $\mon(\xx)$.
  \begin{enumerate}
  \item The set $S_{I,\mord}\coloneqq \setof{u\in \mon(\xx)}{u\notin \lm_{\mord}(I)}$
    is 
    the \emph{staircase} of $I$ w.r.t.\ $\mord$. It naturally
    forms
    a $\field$-vector space basis of $R/I$.
  \item The image of every element $f\in R$ in $R/I$ can be uniquely
    written as a $\field$-linear combination of elements in
    $S_{I,\mord}$. This linear combination of elements in
    $S_{I,\mord}$ is called the \emph{normal form} of $f$ w.r.t.\ $I$
    and $\mord$. The corresponding vector of coefficients of this
    linear combination, with the elements in $S_{I,\mord}$ ordered by
    $\mord$, will be denoted $\NF_{I,\mord}(f)$.
  \end{enumerate}
\end{definition}

We finally define the notion of Gröbner bases.

\begin{definition}
  \label{def:gb} A \emph{Gröbner basis} of an ideal $I\subset R$ w.r.t.\ a
  monomial order $\mord$ is a finite set $G\subset I$ such that
  $\langle \lm_{\mord}(G) \rangle = \lm_{\mord}(I)$. A Gröbner basis is called
  \emph{reduced} if, for any $g\in G$, no monomial in the support of $g$ is
  divisible by any element in $\lm_{\mord}(G\setminus \{g\})$.
\end{definition}

A Gröbner basis $G$ of an ideal $I \subset R$ w.r.t.\ a monomial order
$\mord$ enables the computation of normal forms w.r.t.\ $I$ and
$\mord$ via a straightforward multivariate generalization of
polynomial long division, see e.g.~\cite[Table~5.1]{becker1993}. This,
in particular, yields an ideal membership test for $I$. Indeed, an
element $f\in R$ is contained in $I$ if and only if its normal form
w.r.t.\ $I$ and $\mord$ is zero. Finally, recall that reduced Gröbner
bases are unique for a given ideal and monomial order.

\myvspace{-1em}
\subsection{Points of Good Specialization}
\label{sec:pts}

We start by fixing some notation. For an element $f\in R$ we denote by
$\loc{R}{f}$ the localization of $R$ at the multiplicatively closed
set $\setof{f^k}{k\in \NN}$. For a prime ideal
$\prim\subset R$ we denote by $R_{\prim}$ the localization of
$R$ at the multiplicatively closed set $R\setminus \prim$.

We further fix a polynomial ring $\field[\xz]$ in two finite sets
of variables $\zz$ and $\xx$.  Let $I \subseteq \field[\xz]$ be an
ideal. Suppose that the map
\[\field[\zz]\rightarrow \field[\xz]/I\]
is injective and has \emph{generically finite fiber},
i.e.\ we assume that $\genI \coloneqq I\cdot\field(\zz)[\xx] \neq \field(\zz)[\xx]$ is
a zero-dimensional ideal.
\begin{definition}
  In this setting we call $\genI$ the \emph{generic
    fiber} of $I$.
\end{definition}
Let us introduce some further notation.

\begin{definition}
  \label{def:notation2}
  We denote for a monomial $u\in \mon(\zz)$
  \[\maxi_u\coloneqq \langle v\in \mon(\zz)\;|\;v \drlg u\rangle
    \text{ and }I_u\coloneqq I + \maxi_u,\]
$\maxi \coloneqq \maxi_1 = \langle \zz \rangle$, as well as
$\nxt(u)=\min\setof{v\in \mon(\zz)}{v \drlg u}$.
\end{definition}

\begin{definition}
  \label{def:notation1}
  Let $g\in \field[\zz]_{\maxi}[\xx]$. Write
  \[g = \sum_{w\in \mon(\xx)}\frac{p_w}{q_w}w\]
  with $p_w,q_w\in\field[\zz]$ and $q_w(0)\neq 0$ for all
  $w\in \mon(\xx)$ whenever $p_w\neq 0$. Then, each $p_w/q_w$ can be written
  as a formal power series
  \[\frac{p_w}{q_w} = \sum_{v\in \mon(\zz)}r_{w,v}v\in\field[\![\zz]\!]\]
  and for a monomial $u \in \mon(\zz)$ we denote
  \[g_u \coloneqq \sum_{w\in \mon(\xx)}\sum_{\substack{v\in\mon(\zz)\\v\drl u}}r_{w,v} v
    w = g\bmod\maxi_u\] For a set $G\subset \field[\zz]_{\maxi}[\xx]$ we
  define $G_u\coloneqq \setof{g_u}{g\in G}$.
\end{definition}

Let $G\subset \field(\zz)[\xx]$ be the reduced Gröbner basis of
$\genI$ w.r.t.\ a monomial order $\mordx$ on $\mon(\xx)$.  Our
algorithms will work under the assumption that
$G\subset \field[\zz]_{\maxi}[\xx]$ and that given the set $G_u$ we can lift
$G_u$ uniquely to $G_{\nxt(u)}$. In fact the condition
$G\subset \field[\zz]_{\maxi}[\xx]$ turns out to be sufficient.  We capture
this in a definition:

\begin{definition}
  \label{def:goodspec}
  We say that $\maxi$ is a \emph{point of good specialization} (for
  $\mordx$) if $G\subset \field[\zz]_{\maxi}[\xx]$.
\end{definition}

\begin{remark}
  By definition, being a point of good specialization is a
  Zariski-open condition, so that, if $\field$ is infinite, it is
  ensured with probability $1$ after replacing each $z_i\in \zz$ by
  $z_i-a_i$ for randomly chosen $a_i\in \field$. In
  \Cref{rem:probability} we point out a situation in which an upper
  bound for the probability that $\maxi$ is a point of good
  specialization can be given intrinsically in terms of $I$ if
  $\field$ is finite.
\end{remark}

First we show

\begin{theorem}
  \label{thm:pogs}
  If the ideal $\maxi$ is a point of good specialization, then the
  $\field[\zz]_{\maxi}$-module $\field[\zz]_{\maxi}[\xx]/I$ is free of
  finite rank.
\end{theorem}
\begin{proof}
  Write $A\coloneqq\field[\zz]_{\maxi}$,
  $K = \field(\zz)$ for the field of fractions of $A$ and
  $F\coloneqq A[\xx]/I$.

  Suppose that $\maxi$ is a point of good specialization so that
  $G\subset A[\xx]$. Let $S\coloneqq S_{\genI,\mordx}$, note that $S$ is
  finite. We first show that $S$ generates $F$ as an $A$-module. Let
  $u\in \mon(\xx)$ with $u\notin S$ so that $u\in \lm_{\mordx}(\genI)$. Then
  there exists $g\in G$ and $v\in \mon(\xx)$ such that
  $\lm_{\mordx}(vg) = u$ and therefore such that
  $\lm_{\mordx}(u-v g)\mordx u$.

  Reducing further the expression $u-vg$ by $G$ is done with
  arithmetic over $A$ only and hence shows that, in $F$, we can write
  $u = \sum_{s\in S}r_ss$ with $r_s\in A$. This shows that $S$ generates the
  $A$-module $F$. Now, to prove that $F$ is free over $A$, it suffices
  to show that there are no non-trivial $A$-relations between the
  elements of $S$. Suppose that for certain $s_1,\dots,s_t\in S$, there is a relation
  \[\sum_{i=1}^tr_is_i=0\]
  in $F$ with $r_i\in A\setminus \{0\}$ for all $i$. This gives in particular a
  relation between the $s_i$ over $K$. Hence, if $j$ is such that
  $s_j$ is $\mordx$-maximal among the $s_i$, then $s_j\in L$, but
  $L\cap S = \emptyset$, a contradiction.
\end{proof}

We now give some further properties regarding points of good
specializations. \Cref{enum:lifting:1} in the theorem below will be
used to show the correctness of our lifting algorithm in
\Cref{sec:alg} whereas \Cref{enum:lifting:3} will be used for our
complexity analysis in \Cref{sec:compl}.

\begin{theorem}
  \label{thm:lifting}
  Suppose that $\maxi$ is a point of good specialization.
  \begin{enumerate}
  \item\label{enum:lifting:1} For each $u\in \mon(\zz)$ and
    $z_i\in \zz$, the multiplication by $z_i$ induces an isomorphism of
    $\field$-vector spaces:
    \[u\field[\xz]/I_u\rightarrow z_i u\field[\xz]/I_{z_i u}.\]
  \item\label{enum:lifting:3} Let $\mord$ be the block order eliminating
    $\xx$ with $\mord = \mordx$ on $\mon(\xx)$ and
    $\mord = \drll$ on $\mon(\zz)$.  Let $M_u$ be the (unique) minimal
    generating set of $\maxi_u$. Then, the reduced $\mord$-Gröbner
    basis of $I_u$ is precisely $G_u \cup M_u$.
  \end{enumerate}
\end{theorem}
\begin{proof}
  We reuse the notation from the proof of \Cref{thm:pogs}. By
  \Cref{thm:pogs}, $\maxi$ being a point of good specialization
  implies that $F$ is a free $R$-module of finite rank.

  \emph{Proof of~(\ref{enum:lifting:1})}: It is first easy to check that now
  multiplication by $z_i$ induces a surjective, well-defined map of
  finite-dimensional $\field$-vector spaces
  \[u\field[\xz]/I_u\rightarrow z_i u\field[\xz]/I_{z_i u}.\] Note that the
  structure of $V_u\coloneqq u\field[\xz]/I_u$ as a vector space is
  induced by the canonical $A$-module structure of $F$, because
  $\maxi V_u=0$ and therefore $(V_u)_{\maxi} = V_u$. Hence, if
  $F\cong A^r$, we have,
  \[V_u\simeq (u A/\maxi_u)^r\simeq \field^r.\]
  Therefore multiplication by
  $z_i$ induces an epimorphism between vector spaces of the same
  dimension, so it must be an isomorphism.

  \emph{Proof of~(\ref{enum:lifting:3})}: Let
  $S\coloneqq S_{\genI, \mordx}$. It suffices to show that
  \[S_{I_u,\mord} =S_u \coloneqq \bigcup_{v\drl u}v S.\] Note that the set
  $S_{u}$ certainly generates $\field[\zz,\xx]/(I + \maxi_u)$ as a
  $\field$-vector space. As
  $\field[\zz,\xx]/(I + \maxi_u) \simeq F/\maxi_u$, and since $F$ is free,
  a $\field$-dimension count shows that the set $S_u$ is
  $\field$-linearly independent. Now, let $s\in S$ be
  $\mordx$-minimal such that there exists some $w\in \mon(\zz)$,
  $w\drl u$ with
  $w s \in \lm_{\mord}(I_u)$. By minimality, the $\mord$-normal form of
  $w s$ w.r.t.\ $I_u$ has support in
  $\bigcup_{v\drl u}\setof{v t}{t \mordx s, t\in S}\subset S_u$, therefore
  inducing a linear dependence between the elements of $S_u$, a
  contradiction.
\end{proof}

We will want to perform finite-dimensional linear algebra akin to the
\FGLM algorithm in certain staircases of the ideals $I_u$. This will
rely on the fact that $I_1$ is zero-dimensional.

\begin{corollary}
  \label{prop:zdim}
  Suppose $\maxi$ is a point of good specialization. Then for each
  $u\in \mon(\zz)$, the ideal $I_u$ is zero-dimensional.
\end{corollary}
\begin{proof}
  Note that \Cref{enum:lifting:3} in \Cref{thm:lifting} implies that a
  $\field$-basis of $\field[\xz]/I_u$ is given by
  $\bigcup_{v\drl u}v S_{\genI,\mordx}$ which is a finite set.
\end{proof}

\section{The Main Algorithm}
\label{sec:alg}
As in the last section, we fix an ideal $I\subset\field[\xz]$ with
generically finite fiber over $\field[\zz]$, we will use the notation
introduced in \Cref{def:notation2} and \Cref{def:notation1}. We now
further fix any monomial order $\targmord$ on $\mon(\xx)$ and another
monomial order $\ormord$ on $\mon(\xx\cup \zz)$. We suppose that
$\maxi$ is a point of good specialization for $\targmord$.

Suppose that we can compute, with some black box, the reduced
$\ormord$-Gröbner basis $H_u$ of $I_u$ for any $u\in \mon(\zz)$. Our
goal is to use this data to compute the reduced $\targmord$-Gröbner
basis $G$ of $\genI$.

\begin{remark}
  \label{rem:mis}
  Note that we have so far required that the partition of the
  variables of $\field[\xz]$ is given. It can be computationally
  determined: From any Gröbner basis of $I$ we can determine $\zz$ as
  a \emph{maximally independent set} of $I$ and let $\xx$ be the set
  of remaining variables,
  see~\cite[Definition~9.22]{becker1993}. Then, as in the last
  section, the map $\field[\zz]\rightarrow \field[\xz]/I$ is injective with
  generically finite fiber, see~\cite[Corollary~9.28]{becker1993}.
\end{remark}

Let us sketch our strategy. By the assumption that $\maxi$ is a point
of good specialization, we have $G\subset \field[\zz]_{\maxi}[\xx]$.  Recall
that $\maxi_1=\maxi$. We start by computing the $\ormord$-Gröbner
basis $H_1$ of $I_1=I + \maxi$. Then, we run the \FGLM algorithm
\cite{faugere1993a} with $H_1$ to obtain the reduced
$\targmord$-Gröbner basis of the image of $I$ in
$(\field[\zz]/\maxi)[\xx]\simeq \field[\xx]$. By \Cref{thm:lifting} this
Gröbner basis will now precisely be the set $G_1$ in the notation
introduced in \Cref{def:notation1}.

For a monomial $u\in \mon(\zz)$, let $v \coloneqq \nxt(u)$. Starting
with $u = 1$ and a given $g_1\in G_1$, we will lift $g_u$ to $g_v$ by
performing linear algebra in the finite-dimensional, see \Cref{prop:zdim},
$\field$-vector space $v\field[\xz]/I_{v}$, using the
$\ormord$-Gröbner basis $H_v$. This will rely on \Cref{enum:lifting:1}
in \Cref{thm:lifting}.
\begin{remark}
  \label{rem:tracer}
  In this section we treat the computation of the required Gröbner
  bases $H_u$, $u\in \mon(\zz)$ as a black box. We recall in 
  \Cref{sec:compl} that these sets may be obtained free of arithmetic
  operations from an $\ormord$-Gröbner basis of $I$ when
  $\ormord=\drll$ and $I$ satisfies a certain genericity
  assumption. Under the assumption that $\ormord$ is a suitable block
  order and that $\maxi$ is also a point of good specialization for
  $\ormord$ restricted to $\mon(\xx)$ one can give a {\em
    tracer-based} \cite{traverso1989a} method to compute the sets
  $H_u$.  This will be the subject of a future paper.
\end{remark}

This lifting step is now given by \Cref{alg:lift}.

\begin{algorithm}[]
  \caption{The Lifting Algorithm}
  \label{alg:lift}
  \raggedright

  \begin{description}
  \item[Input] A monomial $u\in \mon(\zz)$, $g_u\in G_u$,
    $v \coloneqq \nxt(u)$, the reduced Gröbner basis $H_v$ of $I_v$
    w.r.t.\ $\ormord$, the set $S_{I_1,\targmord}$.
  \item[Output] The corresponding element $g_v\in G_{v}$.
  \end{description}

  \begin{pseudo}
    \kw{function} \fn{lift}(g_u,H_v,S_{\targmord})\\+
    $c\gets \NF_{I_v,\ormord}(g_u)$ \ct{via $H_v$}\\
    \kw{if} $c = 0$
    \kw{return} $g_u$\\
    \kw{else}\\+
    find $(\alpha_w)_{w\in S_{I_1,\targmord}}$ s.t.\
    $c =\!\!\!\! \sum\limits_{w\in S_{I_1,\targmord}}\!\!\!\!\!\!\alpha_w \NF_{I_v,\ormord}(u w)$
    \ct{via $H_v$}\label{line:lift:vectors}\\
    \kw{return} $g_u-\sum_{w\in S_{I_1,\targmord}}\alpha_w u w$\\
  \end{pseudo}
\end{algorithm}

\begin{theorem}
  \label{thm:corr1}
  If $\maxi$ is a point of good specialization for $\targmord$, then
  \Cref{alg:lift} terminates and is correct in that it satisfies its
  output specification.
\end{theorem}
\begin{proof}
  We use the notation from the pseudocode of the algorithm. The
  termination of the algorithm is clear. For the correctness of the
  algorithm, note that the vectors $\NF_{I_v,\ormord}(u w)$ in 
  \cref{line:lift:vectors} are linearly independent thanks to
  \Cref{enum:lifting:1} of 
  \Cref{thm:lifting}. Thus, there exists at most one choice of
  coefficients $\alpha_w,w\in S_{I_1,\targmord}$, such that
  $c = \sum_{w\in S_{I_1,\targmord}}\alpha_w \NF_{I_v,\ormord}(u w)$.
  Furthermore, since $\maxi$ is a point of good specialization, the
  element $g\in G$ corresponding to $g_u$ provides such a choice of
  coefficients, implying that there exists at least one solution to
  this linear system. This proves the correctness.
\end{proof}
\begin{remark}
  We want to emphasize that our algorithms never verify
  deterministically (and cannot verify deterministically) whether
  $\maxi$ is a point of good specialization, this is a probabilistic
  assumption. Nonetheless, running \Cref{alg:lift} can sometimes
  detect when $\maxi$ is not a point of good specialization, namely if
  there exists no or more than one solution to the linear system in
  \cref{line:lift:vectors} of \Cref{alg:lift}. In this case one would apply a random
  change of coordinates $z_i\gets z_i-a_i$ for each $z_i\in \zz$ and restart
  the computation.
\end{remark}

\begin{example}
  \label{ex:cyclic}
  Let us unroll \Cref{alg:lift} o,n the following example.
  We work over the polynomial ring $\FF_{11}[z,x_1,x_2,x_3]$, where
  $\FF_{11}$ is the finite field with eleven elements. Our ideal $I$ is generated
  by
  \begin{align*}
    &(z+8) + x_1 + x_2 + x_3,\
    (z+8)x_1 + x_1x_2 + x_2x_3 + x_3(z+8),\\
    &(z+8)x_1x_2 + x_1x_2x_3 + x_2x_3(z+8) + x_3(z+8)x_1,\!
    (z+8)x_1x_2x_3 - 1.
  \end{align*}
  Following \Cref{rem:mis}, one verifies that
  $\FF_{11}[z]\rightarrow\FF_{11}[z,x_1,x_2,x_3]/I$ is injective with generically
  finite fiber.

  Readers may recognize this example as the Cyclic~4 polynomial system
  where we have replaced the variable $z$ by the random choice $z+8$
  to ensure, probabilistically, that $\maxi = \langle z \rangle$ is a point of
  good specialization.

  For our orders we choose $\ormord = \drll$ on $\mon(\xx\cup \{z\})$ and
  $\targmord$ as the lexicographic order on $\mon(\xx)$. Now, the set
  $G_1$ is given by
  \[G_1\coloneqq \{x_3^2+6, x_2+8,x_1+x_3\}.\] Hence
  $S_{I_1,\targmord} = \{1,x_3\}$. Assuming that $g_1\coloneqq x_3^2+6$ is the
  image of some element $g$ in the target Gröbner basis
  $G\subset \FF_{11}(z)[\xx]$, we now try to lift $g_1$ to $g_z$, i.e. the
  image of $g$ modulo $\maxi_z = \langle z^2\rangle$, so that, in the notation of
  \Cref{alg:lift}, we have $u = 1$ and $v = z$.

  If such a $g$ exists, there must now exist, by \Cref{enum:lifting:1}
  in \Cref{thm:lifting}, unique scalars $\alpha_1,\alpha_{x_3}\in \FF_{11}$ such that
  \[g_z = g_1 + \alpha_1z + \alpha_{x_3}zx_3 = 0 \mod I_z = I+\langle z^2\rangle,\] and
  \Cref{alg:lift} attempts to compute these scalars by finding a
  linear relation between the normal forms w.r.t. $\ormord$ of
  $g_1,z$ and $zx_3$ modulo $I_z$. Using 
  an
  $\ormord$-Gröbner
  basis of $I_z$, we find that $S_{I_z,\ormord}=\{1,z,x_3,zx_3\}$ and
  we compute, using normal form computations
  \begin{align*}
    \NF_{I_2,\ormord}(g_1) &= (0,7,0,0)\\
    \NF_{I_2,\ormord}(z) &= (0,1,0,0)\\
    \NF_{I_2,\ormord}(zx_3) &= (0,0,0,1)
  \end{align*}
  so that finally, $\alpha_1 = 6$ and $\alpha_{x_3} = 0$ which yields for
  $g_z$ the unique candidate
  $g_{z} = x_3^2 + (4z + 6)$,
  finishing the example.
\end{example}

\Cref{alg:lift} is only able to compute the set $G_u$ for a monomial
$u\in \mon(\zz)$, i.e.\ it ``approximates'' the set $G$ up to order
$u$. A natural question is then how to extract the actual set $G$ out
of $G_u$. For this, we may use the classical technique of Padé
approximants. Having computed the set $G_u$, we have computed the
image $g_u$ of a given element $g\in G$ as
\[ g_u = \sum_{w\in \mon(\xx)}\sum_{v\drl u}r_{w,v} v w.\]
Now we have for the coefficient $p_w/q_w\in\field(\zz)$ of $w$ in $g$
\begin{equation}
  \label{eq:pade}
  p_w - q_w\sum_{v\drl u}r_{w,v}v = 0\bmod \maxi_u,
\end{equation}
which determines a set of linear equations in the unknown coefficients
of $p_w$ and $q_w$. Let $d\coloneqq\deg u$. Suppose that
$\deg \nxt(u) = d + 1$, so that $\maxi_u= \maxi^{d+1}$. Fix $d_1$ and
$d_2$ with $d_1+d_2=d$ and let $n$ be the cardinality of the set
$\zz$. If we impose that $\deg p_w\leq d_1$ and $\deg q_w\leq d_2$, then the
linear system (\ref{eq:pade}) has a finite set of unknowns and
equations. Let us say that any solution to this linear system of
equations is a \emph{Padé approximant of order $(d_1,d_2)$ of
  $\lambda_w\coloneqq\sum_{\deg v<d+1}r_{w,v}v$}. If $d_1$ and $d_2$ are large
enough then any Padé approximant of order $(d_1,d_2)$ of $\lambda_w$ is equal to
$p_w/q_w$, see e.g.~\cite[Proposition 2.1]{guillaume2000}:

\begin{lemma}
  \label{lem:pade}
  Let $p/q$ be a Padé approximant of order $(d_1,d_2)$ of
  $\lambda_w$. If $d_1\geq \deg p_w$ and $d_2\ge \deg q_w$ then
  $p/q = p_w/q_w$.
\end{lemma}

By solving this linear system we obtain an algorithm
$\fn{pade}(g_u,d_1,d_2)$ which computes a candidate
$g_{\cand}\in \field(\zz)[\xx]$ whose coefficients are Padé
approximants of the coefficients of $g_u$ of order $(d_1,d_2)$ regarded
as a polynomial in the variables $\xx$.  Let us say that
$g_u$ has \emph{stable Padé approximation} if for
$v \coloneqq \nxt(u)$ we have
\[g_{\cand} = g_v \bmod \maxi_v.\]
Based on this, we now obtain
\Cref{alg:genfglm} for computing the set $G$ probabilistically. We
state this algorithm in an informal way. In 
\Cref{line:genfglm:lift} by ``lifting $G_{\lift}$ to degree $d$'' we
mean that we compute the set $G_u$ where $u$ is
the $\drll$-maximal monomial of degree $d$.

\begin{algorithm}[]
  \caption{Computing the generic fiber}
  \label{alg:genfglm}
  \raggedright

  \begin{description}
  \item[Input] A generating set $F$ of $I$, a monomial order
    $\ormord$, a monomial order $\targmord$.
  \item[Output] A guess for the set $G$.
  \end{description}

  \begin{pseudo}
    \kw{function} \fn{genfglm}(F,\ormord,\targmord)\\+
    $H_1\gets $ reduced $\ormord$-Gröbner basis of $I_{1}$ \ct{using $F$}\\
    $G_{\lift}, S_{\targmord}\gets \fn{fglm}(H_1,\targmord)$\\
    $G_{\result}\gets \emptyset$\\
    $d \gets 2$\\
    \kw{while} $G_{\lift}\neq \emptyset$\\+
    $G_{\lift} \gets $ lift $G_{\lift}$ to degree $d$
    using 
    \Cref{alg:lift} \label{line:genfglm:lift}\\
    Run $\fn{pade}(g,d/2,d/2)$ for all $g\in G_{\lift}$ \label{line:genfglm:pade}\\
    Lift $G_{\lift}$ one monomial higher  \label{line:genfglm:lift2}\\
    add to $G_{\result}$ all elements with stable Padé approx.\\
    remove the corresponding elements from $G_{\lift}$\\
    $d\gets 2d$\\-
    \kw{return} $G_{\result}$\\
  \end{pseudo}
\end{algorithm}

Clearly, by \Cref{thm:corr1} and \Cref{lem:pade}, this algorithm
returns the correct result if the computed Padé approximants are of
sufficiently large degree and $\maxi$ is a point of good
specialization.

\begin{remark}
  Note that 
  \Cref{alg:lift} works also if we replace $v$ by
  any monomial larger than $u$: In this case we just have to write $c$
  as a linear combination of all the vectors $c_{I_v,\ormord}(u v')$
  where $u \drll v' \drl v$.
\end{remark}

\begin{example}[\Cref{ex:cyclic} continued]
  \label{ex:cycliccont}
  Let us try to see how \Cref{alg:genfglm} recovers the element
  denoted $g$ in \Cref{ex:cyclic}. First, \Cref{alg:genfglm} lifts the
  element $g_1=x_3^2+6$ to degree $d=2$, i.e. we compute the image of
  $g$ modulo $z^3$ (\cref{line:genfglm:lift} 
  in \Cref{alg:genfglm}). This yields, in our
  usual notation,
  \[g_2 = x_3^2 + (2z^2+4z+6).\] Now we attempt
  (\cref{line:genfglm:pade} 
  in \Cref{alg:genfglm}) to find a Padé approximation of order $(1,1)$
  for $g$, i.e., here, $p,q\in \FF_{11}[z]$ of degree at most one such that
  $p/q = 2z^2 + 4z + 6\bmod z^3$ by solving a linear system as outlined
  above.  This yields the candidate
  \[g_{\cand} = x_3^2 + \frac{z+6}{5z+1}\]
  which satisfies $g_{\cand} = g_2\bmod z^3$.
  Next, in 
  \cref{line:genfglm:lift2},
  we lift $g$ one
  monomial higher, i.e. modulo $z^4$. This yields
  \[g_3 = x_3^2 + (7z^3+2z^2+4z+6),\] But $p/q$ has now the truncated
  power series $z^3+2z^2+4z+6$, so that $g_2$ does not have stable Padé
  approximation. Hence we double $d$ to $4$ and lift $g_3$ to $g_4$,
  i.e. from modulo $z^4$ to modulo $z^5$, and attempt another Padé
  approximation. This time, computing a Padé approximation of order
  $(2,2)$, this yields the
  candidate
  \[g_{\cand} = x_3^2 + \frac{1}{10z^2+6z+2}.\]
  Finally,
  we lift $g_4$ to $g_5$ and find that
  $g_5 = g_{\cand} \bmod z^6$.
  So $g_4$ has stable Padé
  approximation and we terminate with $g\coloneqq g_{\cand}$. Computing
  the $\targmord$-Gröbner basis $G$ of $\genI$ using block orders as
  in \cite[Lemma 8.93]{becker1994} shows that $g$ is indeed the
  correct element.
\end{example}

\section{Complexity Estimates}
\label{sec:compl}

In this section, we analyze the arithmetic complexity of a version of
our algorithm more akin to the original \FGLM algorithm as presented
in~\cite{faugere1993a}. We will reuse the notation from the last
section. We now add the additional assumption that $\ormord$ is a
block order eliminating $\xx$ with $\ormord=\drll$ on $\mon(\zz)$ and
that $\maxi$ is a point of good specialization for both $\ormord$ and
$\targmord$. Here, we analyze the number of arithmetic operations in
$\field$ required to obtain the sought $\targmord$-Gröbner basis $G$
using the same strategy as in \Cref{alg:genfglm}, but with a more
optimized lifting step.

Our cost analysis will require measuring the cost of performing
certain linear algebra operations on structured matrices. The matrices
that will appear in the analysis are \emph{block-Toeplitz}:

\begin{definition}
  \label{def:bltoep}
  Let $k,D\in \NN$. A block matrix
  $M= (M_{p q})_{0\leq p,q < k}\in \field^{k D\times k D}$ where each
  $M_{p q}$ is in $\field^{D\times D}$ is called \emph{block-Toeplitz of
  type $(k,D)$} if $M_{p q} = M_{p' q'}$ whenever $p-q = p'-q'$. We say that
  $M$ is \emph{Toeplitz} when $D = 1$.
\end{definition}

If $M$ is block-Toeplitz of type $(k,D)$ then we will need the
arithmetic complexity of computing a matrix vector product $M v$,
$v\in \field^{k D}$, and of inverting $M$. For this we need the concept
of \emph{displacement rank} of a matrix, see
e.g.~\cite{bostan2017a}:

\begin{definition}
  Let $Z\in \field^{n\times n}$ be the matrix defined by
  \[Z \coloneqq 
    (\delta_{i-1,j})_{1\leq i, j\leq n},\]
  where
  $\delta_{i-1,j}$ is the Kronecker delta and let $Z^{\rT}$ be the transpose of $Z$.
  The \emph{displacement rank} of a
  matrix $M\in \field^{n\times n}$ is
  \[\alpha(M)\coloneqq \rk(M - Z M Z^{\rT}).\]
\end{definition}

Note that the displacement rank of a Toeplitz matrix is upper bounded
by $2$. The concept of displacement rank can be used as a general
method to utilize ``Toeplitz-like'' structures in algorithmic linear
algebra. In this vein, we have

\begin{proposition}
  \label{prop:toeplitzcompl}
  Let $M$ be block-Toeplitz of type $(k,D)$ and let $v\in \field^{k D}$.
  Then
  \begin{enumerate}
  \item $M v$ can be computed in $\qcompl{k D^2}$ arithmetic operations
    in $\field$;
  \item $M$ can be inverted in $\qcompl{k D^{\omega}}$.
  \end{enumerate}
\end{proposition}
\begin{proof}
  For any matrix $M\in \field^{n\times n}$, according to~\cite{bostan2017a},
  a matrix-vector product $M v$ can be computed in time
  $\qcompl{\alpha(M)n}$ and $M$ can be inverted in time
  $\qcompl{\alpha(M)^{\omega - 1}n}$. Using a series of rows and column
  swaps, more precisely sending row $p D + i$ to $i k + p$ (resp.\
  column $q D + j$ to $j k + q$), we may transform a block-Toeplitz
  matrix $M$ of type $(k,D)$ into a matrix
  $N = (N_{i j})_{0\leq i,j<D} \in \field^{k D\times k D}$ where each
  $N_{i j}$ lies in $\field^{k\times k}$ and is Toeplitz. Now,
  $N - Z N Z^{\rT}$ has $D$ dense columns and $(k-1)D$ columns with
  potentially nonzero coefficients in positions $i D$ for all
  $i$. Only $D$ of these latter columns can be linearly independent so
  that $\alpha(M)\leq 2 D$, proving both claims.
\end{proof}

Our cost analysis follows closely the one of the original \FGLM
algorithm. To this end, we give the following definition:

\begin{definition}[Multiplication Tensor]
  Let $I$ be a zero-dimensional ideal in a polynomial ring
  $\field[\xx]$ and let $\mord$ be a monomial order. Let
  $S\coloneqq S_{I,\mord}$. The multiplication tensor of $I$ w.r.t.\ $\mord$ is
  defined as the $3$-tensor
  \[M(I,\mord) = (\NF_{I,\mord}(x_i u))_{x_i\in \xx,u\in S},\]
  where the vectors of coefficients are in the basis $S$.
\end{definition}

It turns out that computing these multiplication tensors dominates the
cost of the original \FGLM algorithm and similarly it dominates the
cost of our algorithm. In this section, we denote by $D$ the degree of
$\genI$, i.e.\ the $\field(z)$-dimension of
$\field(z)[\xx]/\genI$. Note that this degree is upper-bounded by that
of $I$. For $u\in \mon(\xx)$ $\drll$-maximal of degree $k$, we also
denote $I_k\coloneqq I_{u}$ and similarly $H_k\coloneqq H_{u}$ and $G_k\coloneqq G_{u}$ with
these sets defined as in the last section.

To simplify the notation, we assume in the following two proofs, that
the set $\zz = \{z\}$ consists of a single variable. It will be clear
from the proofs that they translate accordingly to the more general
setting where $\zz$ consists of several variables.

In our assumed setting, we obtain the following statement for
computing multiplication tensors:

\begin{theorem}
  \label{thm:lifttensor} Let $k\in \NN$ and $\mondeg{k}$ be the number of
  monomials in $\zz$ up to degree $k$. Let $c$ be the cardinality of
  $\xx$. Suppose that we are given the set $H_{2k}$ and the
  multiplication tensor of $I_{k}$ w.r.t.\ $\ormord$. Then the
  multiplication tensor of $I_{2k}$ w.r.t.\ $\ormord$ is computed
  in arithmetic complexity $\qcompl{\mondeg{k} c D^3}$.
\end{theorem}
\begin{proof}
  Let $S \coloneqq S_{I_1,\ormord}$. Applying the third item
  of 
  \Cref{thm:lifting} with $\ormord$ instead of $\mord$, since
  $\ormord$ eliminates $\xx$, for any $\ell\in \NN$, yields
  \[S_{I_\ell,\ormord} = \bigcup_{i=0}^{\ell-1}z^i S.\]

  Let us now describe the structure of the multiplication matrices of
  $I_{2k}$, i.e.\ the matrices
  \[M(I_{2k},\ormord)_y\coloneqq (c_{I_{2 k},\ormord}(y u))_{u\in S_{I_{2 k},\ormord}},\quad\text{for $y\in \xx\cup\zz$}.\]
  The matrix $(c_{I_{2 k},\ormord}(z u))_{u\in S_{I_{2 k},\ormord}}$ of the multiplication by $z$
  is
  \[\bbordermatrix{& S & z S & \dots & z^{2k-1}S \cr
      S & & & &\cr
      z S & \Id & & & \cr
      \vdots & & \ddots & & \cr
      z^{2k-1}S & & & \Id \cr
    }\]
  and so is extracted without any arithmetic operations. Further, denote
  $S_0 \coloneqq \bigcup\limits_{i=0}^{k-1}z^i S$ and $S_1 \coloneqq \bigcup\limits_{i=k}^{2k-1}z^i S = z^k S_0$.
  Now, for
  $x_i\in \xx$ the multiplication matrix $M_{x_i}$ by $x_i$ is determined by
  two matrices $M_{x_i,0},M_{x_i,1}\in \field^{D\times D}$ as follows
  \[M_{x_i}=\bbordermatrix{
      & S_0 & S_1\cr
      S_0 & M_{x_i,0} & \cr
      S_1 & M_{x_i,1} & M_{x_i,0} \cr },
  \]
  where each $M_{x_i,i}$ is easily seen to be block-Toeplitz of type $(\mondeg{k},D)$ and
  $M_{x_i,0}$ is known as part of the $\ormord$-multiplication tensor of
  $I_{k}$. Thanks to the block-Toeplitz structure, it now suffices
  to compute the columns of $M_{x_i,1}$ coming from the normal forms of the
  set $\xx S$, which is of cardinality at most $c D$. Now, we proceed as
  follows: Sort the set $\xx S$ by the monomial order $\ormord$.
  Choose $u\in \xx S$ and suppose that the normal forms of all elements
  less than $u$ in $\xx S$ are known. Two easy cases can arise:
  \begin{enumerate}
  \item $u\in S$, in which case the normal form of $u$ is computed without
    any arithmetic operations;
  \item $u\in \lm(H_{2k})$, in which case the normal form of $u$ is
    computed without any arithmetic operations, it is just given by
    the tail of the corresponding element in $H_{2k}$.
  \end{enumerate}
  Lastly, it can happen that $u\in \lm(I_{2k})$ but
  $u\notin \lm(H_{2k})$. In this case there exists $v\in \xx S$ and
  $x_j\in \xx$ with $u = x_j v$. By assumption the normal form of $v$ is
  known and so is the normal form of each element $x_j b$ with $b\in S$
  and $b\ormord v$. Since $M_{x_j}$ has the same structure as $M_{x_i}$, we
  can now compute the required column of $M_{x_i,1}$ as the sum of two
  matrix-vector products where each of the two matrices is
  block-Toeplitz of type $(\mondeg{k},D)$. This is done in time $\qcompl{\mondeg{k} D^2}$
  thanks to \Cref{prop:toeplitzcompl}, concluding the proof, since
  $\xx S$ has cardinality at most $c D$.
\end{proof}

\myvspace{-0.5em}
Now, we can estimate the complexity of lifting the set $G_{k}$:
\myvspace{-0.5em}
\begin{corollary}
  \label{thm:compl}
  Let $k\in \NN$ and $\mondeg{k}$ be the number of
  monomials in $\zz$ up to degree $k$. Let $c$ be the cardinality of
  $\xx$. Suppose that we are given the set $H_{2k}$, the
  multiplication tensor of $I_{k}$ w.r.t. $\ormord$, the
  $\ormord$-normal forms of the $\targmord$-staircase of $I_1$
  w.r.t. $I_k$ and the $\ormord$-normal forms of the minimal
  $\targmord$-leading monomials of $I_0$ w.r.t. $I_{k}$. Then $G_{2k}$
  is computed in arithmetic complexity $\qcompl{\mondeg{k} c D^{3}}$.
\end{corollary}
\begin{proof}
  By \Cref{thm:lifttensor}, the $\ormord$-multiplication tensor of
  $I_{2k}$ can be computed in arithmetic complexity
  $\qcompl{\mondeg{k} c D^3}$. Having computed this tensor, we proceed
  as follows: let $S$ be the $\ormord$-staircase of
  $I_1 = I+\langle z\rangle$, $T$ be the $\targmord$-staircase of
  $I_1$ and $L$ be the set of minimal $\targmord$-leading terms of
  $I_1$, with $z$ removed. Denoting $S_0$ and $S_1$ as in the proof of
  the preceding theorem, and similarly $T_0$ and $T_1$, now we first
  compute the $\ormord$-normal forms of each element in $T \cup L$
  w.r.t.\ $I_{2k}$, this will yield a tableau of the form
  \[C \coloneqq \bbordermatrix{
      & T_0 & T_1 & L\cr
      S_0 & C_0 & & D_0\cr
      S_1 & C_1 & C_0 & D_1\cr }.
  \]
  Note that by assumption the matrix
  $C_0$ is already known by the $\ormord$-normal forms of $T$
  w.r.t. $I_1$ and the matrix $D_0$ is given by the $\ormord$-normal
  forms of $L$ w.r.t. $I_1$. Note also that $C_0$ and $C_1$ are again
  block-Toeplitz of type $(\mondeg{k},D)$.

  The required matrices $C_1$ and $D_1$ can be computed by using the
  $\ormord$-multiplication tensor of $I_{2k}$ , enumerating the
  monomials in $\mon(\xx)$ in order of $\targmord$ and computing their
  normal forms via matrix-vector multiplications similar to the proof
  of the preceding theorem. Combining this with the block-Toeplitz
  structure of the multiplication matrices of $I_{2k}$ w.r.t.\
  $\ormord$, this can again be done in time $\qcompl{\mondeg{k} c D^3}$.  Finally,
  to compute the set $G_{2k}$, we have to write each column in $C$
  corresponding to an element in $L$ as a $\field$-linear combination
  of the columns corresponding to $T_0\cup T_1$,
  i.e. by solving the linear system
  \[
    \begin{bmatrix}
      C_0 & \cr
      C_1 & C_0\cr
    \end{bmatrix}
    \begin{bmatrix}
      X_0\cr
      X_1
    \end{bmatrix} =
    \begin{bmatrix}
      D_0\cr
      D_1
    \end{bmatrix},
  \]
  where $X_0$ is known via $G_k$. Hence this requires
  \begin{itemize}
  \item inverting the submatrix $C_0$ which, by
    \Cref{prop:toeplitzcompl}, is done in time $\qcompl{\mondeg{k} D^{\omega}}$;
  \item computing the product $C_0^{-1}(D_1-C_1X_0)$.
  \end{itemize}
  Note that $C_0^{-1}$ has displacement rank bounded above by $2D+ 2$,
  see~\cite[Proposition~10.10 and Theorem~10.11]{aecf-2017-livre}, and
  that the cardinality of $L$ is upper bounded by $c D$. Thus, for the
  second step above, we have to compute at most $c D$ matrix-vector
  products of the form $C_0^{-1}v$ and $cD$ matrix-vector products of
  the form $C_1v$. Again, thanks to 
  \Cref{prop:toeplitzcompl}, this is done in time
  $\qcompl{\mondeg{k} c D^3}$. This finally yields the desired
  complexity.
\end{proof}

The following corollary now gives the complexity of computing
successively the sets $G_{2^i}$ from $H_{2^i}$ until $i$ is large
enough to recover $G$, like in \Cref{alg:genfglm}.

\begin{corollary}
  \label{cor:algcompl}
  For $k\in\NN$, let $\mondeg{k}$ be the number of
  monomials in $\zz$ up to degree $k$. Let $c$ be the cardinality of $\xx$.
  Let $\delta-1$ be the maximum degree of all numerators and denominators of
  all coefficients of $G$. Further, let $\ell$ be minimal such
  that $2^{\ell}\geq 2\delta$.

  Given $H_{2^{\ell}}$, computing successively the sets $G_{2^i}$, for
  $i = 1,\dots,\ell$, can be done in arithmetic complexity
  $\qcompl{\mondeg{2^{\ell}} c D^3}=\qcompl{\mondeg{\delta} c D^3}$.
\end{corollary}
\begin{proof}
  Note that the sets $H_{2^i}$, for $i=1,\dots, \ell$, are obtained from
  $H_{2^{\ell}}$ free of arithmetic operations. By \Cref{thm:compl}, the
  computation of $G_{2^i}$ requires $\qcompl{\mondeg{2^i} c D^3}$
  operations. Since
  $\mondeg{2^i}=\binom{n-c+2^i}{n-c}=\compl{2^{(n-c)i}}$, where $n$ is
  the total number of variables, and thus $n-c$ the cardinality of
  $\zz$, summing these complexities for $i$ from $1$ to $\ell$ yields the
  desired complexity.
\end{proof}

Note that going up to degree $2^{\ell}$ suffices to recover the
coefficients of $G$ by Padé approximation thanks to 
\Cref{lem:pade}.

\begin{remark}
  In a 
  follow-up paper, we plan to study the complexity of our algorithm
  using variants of \FGLM,
  such as~\cite{faugere2014a,faugere2017a,neiger2020,berthomieu2022b}.
\end{remark}

We close this section by pointing out a well-known case in which
$\ormord$ is the $\drll$ order, the required Gröbner bases $H_u$ of
$I + \maxi_u$ are extracted without any arithmetic operations of the
$\drll$-Gröbner basis $H$ of $I$ and the $\drll$-staircase of
$I+\maxi_u$ behaves the same as in the above case when $\ormord$ is a
block order. We start with

\begin{definition}
  Let $y$ be
  an extra
  variable and let
  $\homI \subset \field[\xz,y]$ be the homogenization of $I$ w.r.t
  $y$. Suppose that $\homI$ is Cohen-Macaulay. We say that $I$ is in
  \emph{projective generic position} if $\{y\}\cup \zz$ is a maximal
  homogeneous regular sequence in $\field[\xz,y]/\homI$.
\end{definition}

Supposing that $\homI$ is Cohen-Macaulay we now have the following
statement, see e.g.~\cite{lejeune-jalabert1986}.  This statement
has frequently been used in the complexity analysis of Gröbner basis
algorithms.

\begin{lemma}
  \label{lem:fvar}
  Let $I$ be in projective generic position with $\homI$
  Cohen-Macaulay.  Let $H$ be the reduced $\drll$-Gröbner basis of $I$
  (with the variables in $\zz$ considered smaller as those in
  $\xx$). Then
  \[\lm(H) \subset \mon(\xx).\]
  In particular, if $S$ is the $\drll$-staircase of $I_1\coloneqq I + \maxi$,
  then the $\drl$-staircase of $I + \maxi_u$ is given by
  \[S_u\coloneqq \bigcup_{v\drl u}v S.\]
\end{lemma}

This implies that when $I$ is such that $\homI$ is Cohen-Macaulay and
is in projective generic position then we can replace $\ormord$ with
$\drll$ and $H_u$ with $H$ in the statements of 
\Cref{thm:lifttensor} and 
\Cref{thm:compl}. Now we are ready to prove:

\begin{proof}[Proof of {\Cref{th:main}}]
  The genericity assumption on $f_1,\ldots,f_c$ implies that they form a
  Cohen-Macaulay ideal in projective generic position and that the
  ideal has degree $D=d_1\cdots d_c$. Thus, \Cref{alg:genfglm} can be
  called on $\{f_1,\ldots,f_c\}$, $\drll$ and the chosen $\targmord$ in
  order to compute the the reduced $\targmord$-Gröbner basis of
  $I\cdot\field(\zz)[\xx]$ up to precision
  $2\delta$. Finally, using \Cref{thm:compl,cor:algcompl}, we obtain the
  desired complexity.
\end{proof}

\begin{remark}
  \label{rem:probability}
  Let us close this section with a remark on the probability of
  $\maxi$ being a point of good specialization in the situation of
  \Cref{th:main} if $\targmord = \lexl$, the \emph{lexicographic} order
  on $\mon(\xx)$. If $\genI$ is in \emph{shape position}, then the
  reduced $\lexl$-Gröbner basis of $\genI$ is of the form
  \[\{g_c(\zz,x_c), x_1-g_1(\zz,x_c),\dots,x_{c-1}-g_{c-1}(\zz,x_c)\}\]
  with $g_c(\zz,x_c) \in \field[\zz,x_c]$ of total degree $D$. One can
  then show, using \cite{schost2003}, that the degree of the lcm of the
  denominators of the coefficients of $g_1,\dots,g_{c-1}$ is bounded
  by $D^2$ where $D\coloneqq d_1\cdots d_c$ is the Bézout bound of our
  system. Then, if $\field = \FF_q$ for a prime power $q$, the
  Schwartz-Zippel lemma \cite{schwartz1980} implies that the
  probability of $\maxi$ not being a point of good specialization is
  bounded by $D^2/q$ which goes to zero as $q$ increases.
\end{remark}

\myvspace{-0.75em}
\section{Benchmarks}
\label{sec:bench}

In this section, we provide benchmarks for a proof-of-concept
implementation of \Cref{alg:genfglm}. We first give a brief
description thereof.

This implementation is written using the computer algebra system
\oscar~\cite{Oscar2022} which itself is written in
\julia~\cite{bezanson2017}. All required Gröbner basis computations
use the Gröbner basis libraries \msolve~\cite{msolve} via its
\julia-interface \texttt{AlgebraicSolving.jl} or
\texttt{Groebner.jl}~\cite{demin2023}, also written in \julia. The
main step in
\Cref{alg:genfglm}, 
\Cref{alg:lift}, was implemented
naively, close to the provided pseudocode, i.e.\ without the use of
multiplication tensors to compute normal forms as described in 
\Cref{sec:compl}. The implementation is available at
  \url{https://gitlab.lip6.fr/mohr/genfglm}.
For the below benchmarks, the following computations were performed,
keeping the notation from the last sections:
\begin{enumerate}
\item Compute a $\drll$-Gröbner basis for the polynomial ideal $I$ in
  question.
\item Use this $\drll$-Gröbner basis to compute a maximally independent
  set of variables modulo $I$, this gives us the partition of the
  variables into the subsets $\xx$ and $\zz$ as above.
\item If $\zz=\{z_1,\dots,z_{n-c}\}$, choose random $a_1,\dots,a_{n-c}\in \field$
  and make the coordinate substitution $z_i\gets z_i-a_i$.
\item Choose $\ormord$ as the block order on $\mon(\xz)$ eliminating
  $\xx$ with $\drll$ on both blocks of variables.
\item If $\xx=\{x_1,\ldots,x_c\}$, choose $\targmord$ as a block order on
  $\mon(\xx)$ eliminating $\xx'\coloneqq \{x_1,\ldots,x_{c-1}\}$ with $\drll$ on
  $\xx'$ and
  $\prec$
  on $\{x_c\}$.
\item By the elimination property of block orders, the target Gröbner
  basis $G$ contains a single polynomial $g_c$ in the univariate
  polynomial ring $\field(\zz)[x_c]$.
\item Use \Cref{alg:genfglm} to compute only the polynomial $g_c$,
  ignoring the rest of the set $G$. Note that this is indeed possible,
  in \cref{line:genfglm:lift} 
  of \Cref{alg:genfglm} we may choose which of the elements
  in $G_{\text{lift}}$ to actually keep and lift and ignore the rest.
\end{enumerate}
In a certain generic situation (more precisely, when the variable
$x_c$ is ``generic''), the computed polynomial $g_c$ can be used
for a primary decomposition of $I$, see e.g.~\cite[Sections~8.6
and~8.7]{becker1993}, this motivates our choice of $\targmord$.  We
should emphasize however that we did not verify whether this generic
situation is met in the examples below. In the context of primary
decomposition, it suffices to know just the polynomial $g_c$ in
$G$, the rest of the Gröbner basis can be ignored. This is a potential
advantage our algorithm has over the classical way of computing
Gröbner bases of generic fibers using elimination orders for which
there is no way of getting around computing the entire set $G$.

We never directly computed the reduced Gröbner basis $H$ of $I$
w.r.t.\ $\ormord$, but only the reduced $\ormord$-Gröbner basis $H_u$
of the ideals $I + \maxi_u$.  When doing this, we found that the
computations were better-behaved when choosing $\ormord$ as above
rather than $\ormord = \drll$. All computations were performed with
$\field = \mathbb{Z}/p\mathbb{Z}$ where $p$ was a randomly chosen prime of $16$ bits.

We compared the time this computation took with the computation of the
set $G$ using \msolve and which, in this case, just runs the \ffour
algorithm with a suitable block order on $\mon(\xz)$. These timings
are given in 
\Cref{tbl:bench}.

The polynomial systems used for these benchmarks are:
\begin{itemize}
\item ED(3,3) encodes the parametric \emph{euclidean distance} problem
  for a hypersurface of degree 3 in 3 variables, see~\cite{draisma2014};
\item R1, R2, R3 come from a problem in Robotics, see~\cite{garciafontan2022};
\item M2 and M3 are certain jacobian ideals of single multivariate
  polynomials which define singular hypersurfaces;
\item The ``PS'', ``Sing'' and ``SOS'' systems are all critical loci
  of certain projections, see~\cite{eder2023a} for a more detailed
  description;
\item The RD($d$) systems are randomly generated sequences of $3$
  polynomials of degree $d$ in $4$ variables.
\end{itemize}

All computations were performed on an Intel Xeon Gold 6244 CPU @ 3.60
GHz with 1.5 TB of memory. To illustrate the memory usage of our
algorithm, we give in addition the total memory allocated by our
implementation.

\begin{table}[]
  \caption{Benchmarks for 
    \Cref{alg:genfglm}}
  \label{tbl:bench}
  \centering
  \begin{tabular}{l|r|r|r|}
    &\multicolumn{2}{c|}{\Cref{alg:genfglm}} & \msolve with $\targmord$\\[0pt]
    System & Timing (in s) & Memory & Timing (in s)\\[0pt]
    \hline
    ED(3,3) & 237.9 & 95.47 GB & 43521.43\\[0pt]
    R1 & 0.01 & 140.49 GB & 0.01\\[0pt]
    R2 & 0.01 & 251.95 MB & 0.01\\[0pt]
    R3 & 0.01 & 248.93 MB & 0.01\\[0pt]
    M2 & 2.75 & 1.56 GB & 0.03\\[0pt]
    M3 & 0.19 & 410.09 MB & 0.01\\[0pt]
    PS(2,10) & 0.8 & 417.63 MB & 0.3\\[0pt]
    PS(2,12) & 44.82 & 1.38 GB & 7.3\\[0pt]
    Sing(2,10) & 0.2 & 275.17 MB & 0.1\\[0pt]
    SOS(5,4) & 1.2 & 1.2 GB & 0.3\\[0pt]
    SOS(6,4) & 11.97 & 1.07 GB & 30.35\\[0pt]
    SOS(6,5) & 22.19 & 954.41 MB & 26.61\\[0pt]
    RD(3) & 4.29 & 650.98 MB & 0.11\\[0pt]
    RD(4) & 33.42 & 10.46 GB & 13.43\\[0pt]
    RD(5) & 729.51 & 185.87 GB & 780.92\\[0pt]
  \end{tabular}
  \myvspace{-0.25em}
\end{table}

On most small examples in \Cref{tbl:bench}, we achieve a comparable
timing with \msolve, with the exception of PS(2,12). On the larger
example RD(5) we have a small improvement, while on ED(3,3) we achieve
a much better timing. We should mention that our implementation is not
yet optimized (for example, we observed that the linear systems to be
solved in \Cref{alg:lift} are very sparse, yet our implementation
relies on dense representations and arithmetic) and does not
incorporate the idea mentioned in \Cref{rem:tracer}.

\begin{example}
  Continuing \Cref{ex:cycliccont}, we provide a detailed log file with
  explanatory comments for running our implementation on the Cyclic~8
  polynomial system at
  \url{https://polsys.lip6.fr/~mohr/assets/fglm_log.txt}.
\end{example}

\newpage
\bibliographystyle{ACM-Reference-Format}
\bibliography{paper}
\end{document}